\documentclass[twoside]{article}
\usepackage{amsmath,amssymb,amsthm,latexsym,bm,indentfirst,wasysym}
\usepackage{graphicx,epsfig,amscd,mathrsfs,multirow,bm}
\usepackage{cite}
\usepackage{caption}
\usepackage{color}

\usepackage{booktabs}
\usepackage[ruled,linesnumbered]{algorithm2e}

\usepackage{caption}
\usepackage[colorlinks,linkcolor=red,anchorcolor=red,citecolor=blue]{hyperref}
\numberwithin{equation}{section}
\usepackage{titlesec}
\titleformat{\section}{\large\bfseries}{}{0pt}{}

\input xy
\xyoption{all}
\input scrload.tex

\allowdisplaybreaks[4]

\catcode`@=11
\long\def\@makefntext#1{\noindent #1}
\newskip\tabcentering \tabcentering=1000pt plus 1000pt minus 1000pt
\def\MCH#1#2{\setbox0=\hbox{\raise#1\hbox{#2}}\smash{\box0}}

\def\@evenfoot{}\def\@oddfoot{}

\def\@evenhead{\hbox to\textwidth{\small\rm\thepage \hfill
{\it  X. Li, P. Wang}}} 

\def\@oddhead{\hbox to \textwidth{\small{\it
Solve LCS in sequence-to-graph alignment
} \hfill\thepage}}   

\catcode`@=12

\def\bc{\begin{center}}
\def\ec{\end{center}}
\def\no{\noindent}
\def\hang{\hangindent\parindent}
\def\textindent#1{\indent\llap{\qquad #1\ \ \enspace}\ignorespaces}
\def\ref{\par\hang\textindent}

\begin{document}

 \newtheorem{theorem}{Theorem}
 \newtheorem{lemma}{Lemma}
 \newtheorem{corollary}{Corollary}


\abovedisplayskip=6pt plus 1pt minus 1pt \belowdisplayskip=6pt
plus 1pt minus 1pt
\thispagestyle{empty} \vspace*{-1.0truecm} \noindent

\vskip 10mm \bc{\Large\bf 
Polynomial-Time Solutions for Longest Common Subsequence Related Problems Between a Sequence and a Pangenome Graph
\footnotetext{\footnotesize
Supported by Guizhou Provincial Basic Research Program (Natural Science) (Grant No.\,  ZK[2022]020) and the National Natural Science Foundation of China (Grant No. 62362007).\\
* Corresponding author\\  
 E-mail address: xingfuli@mail.gufe.edu.cn
(Xingfu  Li)
} } \ec  

\vskip 5mm
\bc{\bf  Xingfu Li$^{1*}$, Yongping Wang$^{2}$}\\  
{\small\it  $1$. College of Big Data Statistics, Guizhou University of Finance and Economics,  Guiyang $550025$, P. R. China;\\
$2$. Algorithm and Computational Complexity Laboratory of School of Information, Guizhou University of Finance and Economics, Guiyang $550025$, P. R. China
}\ec   

\vskip 1 mm

{\narrower\noindent{\small {\small\bf Abstract}\ \   
A pangenome captures the genetic diversity across multiple individuals simultaneously, providing a more comprehensive reference for genome analysis than a single linear genome, which may introduce allele bias.
  A widely adopted pangenome representation is a node-labeled directed  graph, wherein the paths correspond to plausible genomic sequences within a species. Consequently, evaluating sequence-to-pangenome graph similarity constitutes a fundamental task in pangenome construction and analysis.
  This study explores the Longest Common Subsequence (LCS) problem and three of its variants involving a sequence and a pangenome graph. We present four polynomial-time reductions that transform these LCS-related problems into the longest path problem in a directed acyclic graph (DAG). These reductions demonstrate that all four problems can be solved in polynomial time, establishing their membership in the complexity class P.

\vspace{1mm}\baselineskip 12pt

\no{\small\bf Keywords} \ \ Longest common subsequence; Sequence alignment; Pangenome graph; Algorithm; Complexity 

\par

\vspace{2mm}

\no{\small\bf MR(2020) Subject Classification\ \ {\rm 05C85; 68Q17; 68R10; 92D20; 92-08}} 

}}

\baselineskip 15pt

\section{1. Introduction}
\label{Intro}
A pangenome offers a systematic and inclusive representation of the full genetic diversity across individuals, serving as a more comprehensive and unbiased foundation for genomic studies compared to conventional linear reference genomes, which often introduce allelic bias. As an integrated repository of all gene sequences in a species—including core genes shared by all individuals and accessory genes found only in specific subpopulations—it has become a pivotal resource in computational biology. This framework facilitates high-resolution comparative analyses to uncover genetic variation, infer evolutionary relationships, and explore functional differences among populations, deepening our understanding of biodiversity and adaptive mechanisms \cite{ebler2022pangenome}. The predominant model for representing pangenomes is the vertex-weighted directed graph \cite{baaijens2022computational}, which facilitates sophisticated analyses of genomic variations across populations.

Sequence alignment to a pangenome has emerged as a central endeavor in computational biology, serving to precisely measure the similarity between a query sequence and a comprehensive, species-specific reference pangenome \cite{eizenga2020pangenome}.
When a pangenome is modeled as a directed graph, aligning sequences to the pangenome naturally gives rise to the \emph{sequence-to-pangenome graph alignment problem}, which can be formally defined as the computational challenge of finding an optimal path within a pangenome graph such that the sequence reconstructed along this path achieves the highest possible similarity to a given query sequence \cite{chandra2023gap}.

Aligning a sequence to a pangenome graph is to identify the similarity between them.   Discovering the similarity of biological sequences is a fundamental problem in bioinformatics and has attracted a lot of attention in many applications such as  cancer diagnosis \cite{Aravanis2017NextGenerationSO} and detection of the species common origin \cite{zvelebil2007understanding}. In the field of sequence comparison and analysis, Longest Common Subsequence (LCS) and minimum edit distance are two classical measurement methods. Although both are used to discover the similarity or dissimilarity between sequences, LCS, as an alternative metric, has been adopted to measure the similarities among sequences \cite{banerjee2024longest}, and demonstrates unique advantages in certain scenarios \cite{sahlin2021effective}.  
It is asked to find  the longest common subsequence between two or more sequences. Numerous advanced algorithms have emerged to address the LCS problem among sequences, such as the textbook dynamic programming algorithm \cite{hirschberg1977algorithms}, Fixed-gap LCS \cite{iliopoulos2008algorithms,banerjee2024longest} and longest factor with gaps \cite{banerjee2024longest}.

In this paper,  we rigorously examines the longest common subsequence  between a specified sequence and a pangenome graph, together with three closely associated computational problems.
We establish four reductions respectively from these four problems to longest path problem in a directed-acyclic graph (DAG), which proves that all these four problems are polynomial-time solvable. 

The paper is organized as follows.
Section \ref{pre} introduces basic terminologies, as well as formal definitions of these four LCS-related problems between a sequence and a pangenome graph. At the end of Section \ref{pre}, 
we propose an efficient algorithm to address the longest path problem in a DAG, serving as a fundamental sub-procedure for tackling the four LCS-related problems explored in this study.
In the subsequent Section \ref{lower-bound}, we demonstrate that finding a longest common subsequence between a sequence and a pangenome graph can not be addressed in sub-quadratic time.  
We then transform each of the four LCS-related problems into the problem of finding the longest path in a specific DAG, as detailed in Section~\ref{reduction}. 
We analyze the time complexity of solving these problems via reductions in Section \ref{complexity-analysis}.
Finally, we conclude this paper and propose some future work in Section \ref{conclusion}.

\section{2. Preliminary}
\label{pre}
The \emph{length} of a sequence $S$, denoted by $|S|$, refers to the total number of elements it contains. Indices, represented as subscripts, indicate the positions of elements within the sequence, and in this work, all indexing starts from zero. We define $S[i,\dots, j]$ as the substring of $S$ spanning from index $i$ to index $j$, inclusive. For instance, consider the sequence $S$ = "$ababc$"; the leftmost '$a$' is located at position $0$, whereas the rightmost '$a$' appears at $2$. Furthermore, the substring $S[1,\dots, 4]$ yields "$babc$".

A \emph{pangenome graph} $G = (V, E, \delta)$ is a directed graph with vertex labeled, where $V$ denotes the set of vertices, $E$ represents the set of directed edges, and $\delta: V \rightarrow \Sigma^* \setminus \{\epsilon\}$ is a labeling function that assigns to each vertex a non-empty finite string over a given alphabet $\Sigma$, with $\epsilon$ denoting the empty string.
In this paper,  the sequence associated with  a vertex $v\in V(G)$ is written as $\delta(v)$. 
A path $P = u_0, u_1, u_2, \dots, u_k$ in a directed graph $G$ is defined as a sequence of vertices where each consecutive pair $(u_i, u_{i+1})$ is connected by a directed edge for all $0 \leq i \leq k-1$. The sequence $spell(P) = \delta(u_0)\delta(u_1)\dots\delta(u_k)$, referred to as the \emph{spell} of path $P$, is formed by concatenating the vertex labels along the path in order. More precisely, in a pangenome graph, the sequence spelled by a path is obtained by successively joining the sequence associated with each vertex as they appear along the path. We say that vertex $u$ reaches vertex $v$, or equivalently that $v$ is reachable from $u$, if there exists a directed path from $u$ to $v$ in $G$.

\subsection{Longest common subsequence related problems between a sequence and a pangenome graph}
\label{problems}
Given a sequence $Q$ and a pangenome graph $G$, if there is a subsequence $S$ 
and a path $P$ in $G$ such that $S$ is a common subsequence between  $spell(P)$ and $Q$, then $S$ is said to be a common subsequence between $Q$ and $G$. 

\textbf{Problem 1} (Longest Common Subsequence  between a Sequence and a pangenome Graph (LCS-SG) problem). Given a sequence $Q$ and a pangenome graph $G$,   find a longest common subsequence between $Q$ and $G$. 

Let $S$ denote a common subsequence between a sequence $Q$ and a pangenome graph $G$. Define $\text{ind}(S_i, Q)$, where $0 \leq i \leq |S| - 1$, as the index of element $S_i$ in $Q$. For instance, consider $Q = "xyabcahde"$ and $S = "acah"$. Then, $\text{ind}(S_0, Q) = 2$, $\text{ind}(S_1, Q) = 4$, $\text{ind}(S_2, Q) = 5$, and $\text{ind}(S_3, Q) = 6$.
A subsequence $S$ of $Q$ 
is said to be a \emph{$(k_1,k_2)$-gap common subsequence} between $Q$ and $G$ if there is a path $P$ in $G$ such that the following conditions hold: (1) $S$ is a common subsequence between $spell(P)$ and $Q$; (2) $0<ind(S_i,Q)-ind(S_{i-1},Q)\leq k_1$  for every $1\leq i\leq |S|-1$; and (3) $0<ind(S_i,spell(P))-ind(S_{i-1},spell(P))|\leq k_2$  for every $1\leq i\leq |S|-1$.

\textbf{Problem 2} (Fixed Gap Longest Common Subsequence between a Sequence and a pangenome Graph (FGLCS-SG) problem).  Given a sequence $Q$, a pangenome graph $G$ and two additional integers $k_1,k_2>0$. Its aim is to find a longest $(k_1,k_2)$-gap common subsequence  between $Q$ and $G$. 

Let $Q$ be a sequence and $G=(V,E,\delta)$ be a pangenome graph. A triple, $(v,[i_h,\dots,i'_h], [j_h,\dots,j'_h])$, is said to be an \emph{exact match} (or simply a \emph{match} or \emph{seed}) between $Q$ and $G$, if $\delta(v)[i_h,\dots, i'_h]$ $=$ $Q[j_h,\dots,j'_h]$, where $v\in V(G)$.  A match $(v, [i_h,\dots,i'_h],$ $[j_h,\dots,j'_h])$ is \emph{left-maximal} if $\delta(v)_{i_h-1}$ $\not =$ $Q_{j_h-1}$ or $i_h$ $=$ $0$ or $j_h$ $=$ $0$. A match $(v, [i_h,\dots,i'_h],$ $[j_h,\dots,j'_h])$ is \emph{right-maximal} if $\delta(v)_{i'_h+1}$ $\not =$ $Q_{j'_h+1}$ or $i'_h$ $=$ $|\delta(v)|-1$ or $j'_h$ $=$ $|Q|-1$. If a match is both left and right maximal,  it is a Maximal Exact Match (MEM). 

For a given match $m = (v, [i_h,\dots,i'_h], [j_h,\dots,j'_h])$, we define its components as follows: $m_0$ represents the vertex $v$, $m_1$ denotes the sequence of indices $[i_h,\dots,i'_h]$, and $m_2$ corresponds to the sequence $[j_h,\dots,j'_h]$. Additionally, the notations $|m_1|$ and $|m_2|$ are used to indicate the number of indices contained in $m_1$ and $m_2$, respectively, effectively capturing the length or size of each index list. Note that the length of $m_1$ is equal to that of $m_2$, i.e., $|m_1|$ $=$ $|m_2|$, since $m$ is an exact match. 

Two seeds $(v_1, [i_1,\dots,i'_1], [j_1,\dots,j'_1])$ and $(v_2, [i_2,\dots,i'_2], [j_2,\dots,j'_2])$, representing local exact matches between a query sequence $Q$ and a pangenome graph $G$, are said to be \emph{strictly ordered} or \emph{in strict order} when  $j'_1 < j_2$, and one of the following conditions is satisfied: if the two seeds map to the same node in the graph $G$, i.e.,  $v_1 = v_2$, then $i'_1 < i_2$; otherwise, if $v_1 \neq v_2$, the node $v_1$ must be able to reach $v_2$ via a valid path in the directed structure of the pangenome graph $G$, ensuring a consistent genomic order is preserved.
The two strictly ordered seeds are denoted as $(v_1,$$[i_1,\dots,i'_1],$ $[j_1,\dots,j'_1])$ $\prec$ $(v_2,$$[i_2,\dots,i'_2],$ $[j_2,\dots,j'_2])$ throughout this paper.

Let $\mathcal{M}$ denote a subset of MEMs between a sequence $Q$ and a pangenome graph $G$.
 The length of $\mathcal{M}$ is defined as  $len(\mathcal{M})$ $=$ $\sum_{m\in\mathcal{M}}|m_1|$ $=$ $\sum_{m\in\mathcal{M}}|m_2|$. If the seeds in $\mathcal{M}$ can be arranged linearly such that every pair of adjacent seeds is strictly ordered, then $\mathcal{M}$ is termed a \emph{strictly ordered subset}.

 \textbf{Problem 3} (MEM Chain (MEMC) problem). Given a subset of MEMs, $\mathcal{M}$,  and a pangenome graph $G=(V,E,\delta)$, find a strictly ordered subset $\mathcal{M}'$  from $\mathcal{M}$ such that $len(\mathcal{M}')$ is maximized. 

 If we wish to find a path in a pangenome graph so that the path can traverse through as many matches as possible, rather than care about the length of sequence associated with each vertex, we shall come up with the Max Seed Problem. 

\textbf{Problem 4} (Max Seed Problem (MSP)). Given a sequence $Q$; a pangenome graph $G=(V,E,\delta)$; a set of seeds $\mathcal{S}$  between $Q$ and $G$.  MSP asks to find  a path in $G$ so that the path passes through maximum number of  seeds from $\mathcal{S} $ in a strict order.

\subsection{Longest path on DAGs}
\label{LP-DAG}
A Directed Acyclic Graph (DAG) is a graph where every edge has a direction, and no cycles exist. This means that starting from any vertex and traversing along the directed edges, one can never return to the starting point. DAGs possess numerous intriguing properties and are widely used to represent real-world structures, such as task dependencies and version control systems. 

In a DAG, the \emph{edge-weighted} longest path problem entails identifying the path of maximum length from any source vertex to any sink vertex, where the path length is defined as the sum of its edge weights. Although computing shortest paths in cyclic graphs can be efficiently achieved—e.g., via Dijkstra’s algorithm \cite{cormen2022introduction}—the longest path problem in such graphs is NP-hard \cite{garey2002computers}. Fortunately, the absence of cycles in DAGs enables an efficient solution to the longest path problem.

The key to solving the longest path problem in a DAG lies in combining two powerful techniques: topological sorting, which orders vertices to respect edge directions, and dynamic programming, which builds solutions incrementally by reusing prior results.

Topological sorting refers to a linear ordering of vertices in a DAG such that for every directed edge $(u, v)$, vertex $u$ precedes vertex $v$ in the sequence. This ordering is feasible exclusively for DAGs, as cycles would prevent a valid order, and notably, the topological order is not unique—multiple valid sequences may exist. Topological sorting on a given DAG $G$ can be efficiently computed in $O(|V(G)| + |E(G)|)$ time  using the following procedure \cite{cormen2022introduction}. First, perform a depth-first search on $G$ to determine the finishing-visited time for each vertex of $G$. Then, after each vertex finishes, insert it onto the front of a linked list. The resulting  list is a topological sorting vertices. The correctness of this method, along with its linear-time complexity, is  established in \cite{cormen2022introduction}.

Let us explore how dynamic programming can be effectively employed to determine the longest path. The core idea involves processing vertices in topological order while dynamically tracking the length of the longest path terminating at each vertex.
Let $Q$ be a topological sorting vertex list of a given DAG $G$. Let $dist^{e}(v)$ be  the edge-weighted longest path length ending at the vertex $v$. A vertex $u$ is an \emph{in-neighbor} of another vertex $v$ in $G$ if there is a directed edge from $u$ to $v$.
Let $N^{-}(v)$ be the set of in-neighbors of vertex $v$ in the graph $G$. 
Since each edge in $G$ follows the order of $Q$, the iterative relationship \ref{edge-weighted-relation} holds for every index $i: 0\leq i <|Q|-1$ of $Q$: 
\begin{equation}\label{edge-weighted-relation}
    dist^{e}(Q_i) = \begin{cases}
                0                                               & N^{-}(Q_i) = \emptyset\\
                \max_{u\in N^{-}(Q_i)}\Big\{dist^{e}(u) + weight\big((u,Q_i)\big)\Big\} & N^{-}(Q_i)\not=\emptyset
               \end{cases},
\end{equation}
where the index of $u\in N^{-}(Q_i)$ for $1\leq i <|Q|-1$ must be less than $i$ in $Q$.

The complete algorithm for solving the edge-weighted longest path problem on a DAG is outlined in Algorithm \ref{edge-longest-alg}. As this algorithm processes each vertex and edge a constant number of times, its time complexity is $O\big(|V(G)| + |E(G)|\big)$ for any given DAG $G$.

\begin{algorithm}
    \caption{Edge-weighted Longest path on a DAG\label{edge-longest-alg}}
    \KwIn{A directed-acyclic graph $G=(V,E)$ in which each edge $(u,v)\in E$ has non-negative weight $weight\big((u,v)\big).$}
    \KwOut{An edge-weighted longest path in $G$.}
    Perform a topological sort of $G$\;
    Initialize  arrays $dist$ and $\pi$ respectively to store  longest path lengths and the parent vertex in the longest path ending at each vertex\; 
    Set $dist[v] = 0$ and $\pi[v] = \emptyset$ for each vertex $v$ initially\;
    \For{each vertex $u$ in the topological order}
    {
        \For{each out-neighbor $v$ of $u$ \tcc*{A vertex $v$ is an out-neighbor of $u$ if there is a directed edge from $u$ to $v$.}} 
        {
            \If{$dist[u] + weight\big((u,v)\big) > dist[v]$}
            {
            Update $dist[v] =  dist[u] + weight\big((u, v)\big)$\;
            $\pi[v]$ $\leftarrow$ $u$\;
            }

        }
    }
    \Return the maximum value in $dist$ array and the list $\pi$.
\end{algorithm}

 Given a DAG where each vertex is assigned a non-negative weight, the vertex-weighted longest path problem seeks to identify a path that maximizes the sum of the vertex weights along the path. Let $dist^{v}(u)$ denote the maximum total vertex weight of any path ending at vertex $u$, representing the longest vertex-weighted path length terminating at $u$.

 The longest path length ending at each vertex in the vertex-weighted longest path problem exhibits an iterative relationship in equation \ref{vertex-weighted-relation} which is similar to the edge-weighted one  expressed in equation \ref{edge-weighted-relation}.
\begin{equation}\label{vertex-weighted-relation}
    dist^{v}(Q_i) = \begin{cases}
                weight(Q_i)                                               & N^{-}(Q_i) = \emptyset\\
                \max_{u\in N^{-}(Q_i)}\big\{dist^{v}(u) + weight(Q_i)\big\} & N^{-}(Q_i)\not=\emptyset
               \end{cases}
\end{equation}

Algorithm \ref{ver-longest-alg} presents the complete procedure for solving the vertex-weighted longest path problem in a DAG. Analogous to Algorithm \ref{edge-longest-alg}, it achieves a time complexity of $O\big(|V(G)| + |E(G)|\big)$ for a given DAG $G$, making it efficient in practice.

\begin{algorithm}
    \caption{Vertex-weighted Longest path on a DAG\label{ver-longest-alg}}
    \KwIn{A directed-acyclic graph $G=(V,E)$ in which each vertex $v\in V$ has non-negative weight $weight(v).$}
    \KwOut{A vertex-weighted longest path in $G$.}
    Perform a topological sort of the graph\;
    Initialize  arrays $dist$ and $\pi$ respectively to store  longest path lengths and  parent vertex in the longest path terminating at each vertex\; 
    Set $dist[v] = weight(v)$ and $\pi[v] = \emptyset$ for each vertex $v$ initially\;
    \For{each vertex $u$ in the topological order}
    {
        \For{each out-neighbor $v$ of $u$ \tcc*{A vertex $v$ is an out-neighbor of $u$ if there is a directed edge from $u$ to $v$.}} 
        {
            \If{$dist[u] + weight(v) > dist [v]$}
            {
            Update $dist[v] = dist[u] + weight(v)$\;
            $\pi[v]$ $\leftarrow$ $u$\;
            }
        }
    }
    \Return the maximum value in $dist$ array and the list $\pi$.
\end{algorithm}

\section{3. A lower bound for the longest common subsequence problem between a sequence and a pangenome graph}
\label{lower-bound}
The longest common subsequence (LCS) problem between two sequences is a classic challenge in computer science and bioinformatics, known to lack a sub-quadratic time algorithm \cite{bringmann2015quadratic}, under the assumption of the Strong Exponential Time Hypothesis (SETH), which suggests that in practice no significantly faster algorithm exists for solving it exactly. This computational hardness has important implications for applications requiring efficient sequence comparison, such as genome analysis and text processing. 

Given two sequences $Q_1, Q_2 \in \Sigma^{*}$ over an alphabet $\Sigma$, where $\Sigma^{*}$ denotes the set of all possible strings formed from $\Sigma$, we construct a pangenome graph $\hat{G}$ comprising a single vertex $v$ and assign the label $\delta(v) = Q_2$, effectively encoding one of the sequences directly into the graph structure. By aligning $Q_1$ against this simplified graph representation, we transform the traditional LCS problem into a graph-based variant, thereby yielding an instance of the LCS-SG  problem, which extends the classical formulation to more complex genomic representations.

\begin{lemma}
    \label{LCS-preserve-opt}
    There is an LCS between $Q_1$ and $Q_2$ with length at least $\ell$ if and only if there is an LCS between $Q_1$ and $\hat{G}$ with length at least $\ell$.
\end{lemma}
\begin{proof}
    Since the graph $\hat{G}$ contains only a single vertex, and the sequence associated with this vertex is one of the two sequences $Q_2$, a common subsequence between $Q_1$ and $Q_2$ is equivalent to a common subsequence between $Q_1$ and $\hat{G}$, and vice versa. Thus, this lemma is valid.
\end{proof}

\begin{theorem}
    \label{LCS-SG-lower-bound}
    The LCS-SG problem can not be solved in sub-quadratic time conditioned on SETH. 
\end{theorem}
\begin{proof}
    Let $Q_1$ and $Q_2$ be two sequences over an alphabet. The construction of $\hat{G}$ begins by creating a single vertex and then assigning $Q_2$ as the sequence associated with this vertex. There is not any edge in $\hat{G}$. So the construction takes totally linear-time.  
    Furthermore, by virtue of Lemma \ref{LCS-preserve-opt}, the validity of this theorem is rigorously established.
\end{proof}

\begin{corollary}
    \label{FGLCS-SG-lower-bound}
    The FGLCS-SG problem can not be solved in sub-quadratic time conditioned on SETH.
\end{corollary}
\begin{proof}
    In instances of the LCS-SG problem, we further introduce two parameters $k_1 = +\infty$ and $k_2 = +\infty$, effectively transforming it into the FGLCS-SG problem. In other words, the LCS-SG problem constitutes a sub-problem of the FGLCS-SG problem, distinguished by the absence of gap constraints in the former.  According to Theorem~\ref{LCS-SG-lower-bound}, the FGLCS-SG problem cannot be solved in sub-quadratic time, establishing a fundamental lower bound for its computational complexity.
\end{proof}

\section{4. Problem transformations}
\label{reduction}
First, we investigate the LCS-SG problem.
We construct a new directed graph $H$ as follows:

The vertex set is defined by
$V(H) = \big\{ v_{(i,u,f)} \mid 0 \leq i \leq |Q| - 1,\, u \in V(G),\, \text{and}~Q_i = \delta(u)_f~\text{for some}~0 \leq f \leq |\delta(u)|-1 \big\}$. 

The edge set is given by
$E(H) = \big\{ (v_{(i,u,f)}, v_{(i',u',f')}) \mid 0 \leq i < i' \leq |Q| - 1,\, u, u' \in V(G),\, 0 \leq f \leq |\delta(u)|-1,\, 0 \leq f' \leq |\delta(u')|-1,~\text{and}$
$ \text{either}~u \ne u'~\text{and there exists a path from}~u~\text{to}~u'~\text{in}~G,, \text{or}~u = u'~\text{and}~f < f'$.$\big\}$

\begin{lemma}\label{H-is-DAG}
    $H$ is a DAG.
\end{lemma}
\begin{proof}
Suppose there is a cycle $C = v_{(i_0,u_0,f_0)}, v_{(i_1,u_1,f_1)},v_{(i_2,u_2,f_2)},\dots,v_{(i_k,u_k,f_k)}, v_{(i_0,u_0,f_0)}$. Based on the construction of $H$, we obtain $i_0 < i_1 < i_2 < \dots < i_k  <i_0$, which implies $i_0 < i_0$, yielding a contradiction. Hence, $H$ must be a DAG.
\end{proof}

\begin{lemma}\label{H-path-to-CS}
    A directed path in $H$ indicates a common subsequence between $Q$ and $G$. 
\end{lemma}
\begin{proof}
    Let $P=u_{(i_0,u_{0},f_{0})},u_{(i_1,u_{1},f_{1})},\dots,u_{(i_k,u_{k},f_{k})}$ be a directed path in $H$. According to the construction of  $H$, for every vertex $u_{(i_j,u_{j},f_{j})}\in V(P): 0\leq j\leq k$, we have $Q_{i_j} = \delta(u_{j})_{f_{j}}$. 
    In addition, for every pair $\big($$u_{(i_j,u_{j},f_{j})},u_{(i_{j+1},u_{j+1}, f_{j+1})}\big):$ $0\leq j <k$, we have $i_j < i_{j+1}$; if $u_j\not =u_{j+1}$, then $u_j$ can reach $u_{j+1}$; otherwise, $f_j<f_{j+1}$. 
    So these indicates that   the sequence $Q_{i_0},Q_{i_1},\dots,Q_{i_k}$ is a common subsequence between $Q$ and $G$.
\end{proof}

\begin{lemma}\label{H-CS-to-path}
    Every common subsequence between $Q$ and $G$ indicates  a directed path of $H$.
\end{lemma}
\begin{proof}
     Let $S$ be a common subsequence between $Q$ and $G$. Let $\rho(i) = ind(S_i, Q)$  
     for $0\leq i\leq |S|-1$. Since $S_i$ must be in front of $S_{i+1}$ for $0\leq i \leq |S|-2$ in the sequence $S$, the order in which they appear on $Q$ remains the same as that on $S$, i.e., $\rho(i)$ $<$ $\rho(i+1)$.

    Let $P$ be the path in $G$ such that $S$ is a common subsequence between $spell(P)$ and $Q$.  Let $X=spell(P)$. Let $\phi(i) = ind(S_i, X)$ for $0\leq i\leq |S|-1$.
    Let $u_{\phi(i)}\in V(G):0\leq i\leq |S|-1$ be the vertex from which $X_{\phi(i)}$ comes. Let $f_{\phi(i)}$ $=$ $ind\big(X_{\phi(i)},\delta(u_{\phi(i)})\big)$ for $0\leq i\leq |S|-1$.

    If $u_{\phi(i)}=u_{\phi(i+1)}$, then $f_{\phi(i)}<f_{\phi({i+1})}$ since $X_{\phi(i)}$ is in front of $X_{\phi({i+1})}$ in sequence $X$.
    Thus, in the graph $H$, there is a directed edge from $v_{(\rho(i),u_{\phi(i)},f_{\phi(i)})}$ to the vertex $v_{(\rho({i+1}),u_{\phi({i+1})},f_{\phi({i+1})})}$ for every $0\leq i\leq |S|-2$.

    If $u_{\phi(i)}\not =u_{\phi({i+1})}$, then $u_{\phi(i)}$ can reach $u_{\phi({i+1})}$ through the path $P$. Thus,  in the graph $H$, there is a directed edge from $v_{(\rho(i),u_{\phi(i)},f_{\phi(i)})}$ to the vertex $v_{(\rho({i+1}),u_{\phi({i+1})},f_{\phi({i+1})})}$ for every $0\leq i\leq |S|-2$.

    Therefore,  $v_{(\rho(0),u_{\phi(0)},f_{\phi(0)})}$, $v_{(\rho(1),u_{\phi(1)},f_{\phi(1)})}$,$\dots$,$v_{(\rho({|S|-1}),u_{\phi({|S|-1})},f_{\phi({|S|-1})})}$ is a directed path in $H$. 
\end{proof}

\begin{theorem}\label{H-longest-path-is-LCS}
    A longest path in $H$ indicates a longest common subsequence between $Q$ and $G$. 
\end{theorem}
\begin{proof}
     By Lemma \ref{H-CS-to-path}, every  common subsequence between $Q$ and $G$ is indicated by a directed path in $H$. On the other hand, each directed path in $H$ indicates a common subsequence between $Q$ and $G$ according to Lemma \ref{H-path-to-CS}. So a longest path in $H$ must derive a longest common subsequence between $Q$ and $G$.
\end{proof}

Now let us consider the FGLCS-SG problem. 
We construct a new directed graph $H_{gap}$ as follows:

The vertex set is defined by
$V(H_{gap})$ $=$ $\big\{$ $v_{(i,u,f)}$:  $0\leq i\leq |Q| - 1$, $u\in V(G)$ such that $Q_i = \delta(u)_{f}$ for some $0\leq f\leq |\delta(u)|-1$ $\big\}$.

The edge set is given by
$E(H_{gap})$ $=$ $\big\{$ $(v_{(i,u,f)}, v_{(i',u',f')})$: $0<i'-i\leq k_1$, $u, u'\in V(G)$,  $0\leq f< |\delta(u)|$,   $0\leq f' < |\delta(u')|$  and if $u\not =u'$, then there is a path $P$ from $u$ to $u'$ in $G$ such that $0<ind\big(\delta(u')_{f'},spell(P)\big) - ind\big(\delta(u)_{f},spell(P)\big)\leq k_2$; otherwise $0<f'-f\leq k_2$ $\big\}$.

\begin{lemma}\label{H-gap-is-DAG}
    $H_{gap}$ is a DAG. 
\end{lemma}
\begin{proof}
    Comparing $H_{gap}$ and $H$, it is conducted to explore their structural characteristics. Both of them share an identical vertex set, and the edge set of $H_{gap}$ is contained within $E(H)$, establishing $H_{gap}$ as a subgraph of $H$.
    So $H_{gap}$ is also a DAG.
\end{proof}

\begin{lemma}\label{H-gap-path-to-CS}
    A directed path in $H_{gap}$ indicates a $(k_1,k_2)$-common subsequence between $Q$ and $G$. 
\end{lemma}
\begin{proof}
     By Lemma \ref{H-gap-is-DAG}, $H_{gap}$ is  a DAG and a subgraph of $H$.  Then by Lemma \ref{H-path-to-CS}, every directed path in $H_{gap}$ indicates a common subsequence between $Q$ and $G$.  In addition, according  to the construction of  $H_{gap}$, a directed path in $H_{gap}$ indicates a $(k_1,k_2)$-common subsequence between $Q$ and $G$. 
\end{proof}

\begin{lemma}\label{H-gap-CS-to-path}
    Every $(k_1,k_2)$-common subsequence between $Q$ and $G$ indicates a directed path of $H_{gap}$.
\end{lemma}
\begin{proof}
     Let $S$ be a $(k_1,k_2)$-common subsequence between $Q$ and $G$. Then there is a path $P$  in $G$ such that $S$ is a common subsequence between $spell(P)$ and $Q$. Let $X = spell(P)$. For  $0\leq i\leq |S|-1$ Let $\phi(i) = ind(S_i, X)$.
     In addition, let $\rho(i)$ $=$ $ind(S_i, Q)$ for $0\leq i\leq |S|-1$. 
      Let $u_{\phi(i)}$ be the vertex in $G$ which the element $X_{\phi(i)}$ comes from, and $f_{\phi(i)}$  $=$ $ind\big(X_{\phi(i)}, \delta(u_{\phi(i)})\big)$.
    
    For every two adjacent elements $S_i$ and $S_{i+1}$, where $0\leq i\leq |S|-2$, since $S$ is a $(k_1,k_2)$-common subsequence between $Q$ and $G$, we have $\rho(i+1)$ $-$ $\rho(i)$ $\leq$ $k_1$ and $\phi(i+1)$ $-$ $\phi(i)$ $\leq$ $k_2$. Now we shall prove that there is an edge from  $v_{(\rho(i), u_{\phi(i)}, f_{\phi(i)})}$ to $v_{(\rho(i+1), u_{\phi(i+1)}, f_{\phi(i+1)})}$  for $0\leq i< |S|-1$.

    If $u_{\phi(i)}$ $=$ $u_{\phi(i+1)}$, then $X_{\phi(i)}$ and $X_{\phi(i+1)}$ are on the same vertex. Then $f_{\phi(i)} = \phi(i)$ and $f_{\phi(i+1)} = \phi(i+1)$.  Since $\phi(i+1)$ $-$ $\phi(i)$ $\leq$ $k_2$, then we have $f_{\phi(i+1)} - f_{\phi(i)}$ $\leq$ $k_2$. 
    If $u_{\phi(i)}$ $\not=$ $u_{\phi(i+1)}$, since $u_{\phi(i)}$ and $u_{\phi(i+1)}$ are both on the path $P$, $u_{\phi(i)}$ can reach  $u_{\phi(i+1)}$ in $G$. 

    So there is an edge from  $v_{(\rho(i), u_{\phi(i)}, f_{\phi(i)})}$ to $v_{(\rho(i+1), u_{\phi(i+1)}, f_{\phi(i+1)})}$ for $0\leq i< |S|-1$. Thus $S$ indicates a directed path of $H_{gap}$.
\end{proof}

\begin{theorem}\label{H-gap-longest-path-is-LCS}
    A longest path in $H_{gap}$ indicates a longest $(k_1,k_2)$-common subsequence between $Q$ and $G$. 
\end{theorem}
\begin{proof}
     By Lemma \ref{H-gap-CS-to-path}, every  $(k_1,k_2)$-common subsequence between $Q$ and $G$ indicates a directed path in $H_{gap}$. On the other hand, each directed path in $H_{gap}$ indicates a $(k_1,k_2)$-common subsequence between $Q$ and $G$ according to Lemma \ref{H-gap-path-to-CS}. So a longest path in $H_{gap}$ must derive a longest $(k_1,k_2)$-common subsequence between $Q$ and $G$.
\end{proof}

Now Let us consider MEMC problem. Recall that the instance of MEMC is a subset $\mathcal{M}$ of MEMs and an additional pangenome graph $G$. Its aim is to find a strictly ordered subset $\mathcal{M}'$ from $\mathcal{M}$ so that $len(\mathcal{M}')$ is maximum. We construct a vertex-weighted directed graph $H_{MEM}$ as following:
$V(H_{MEM})$ $=$ $\{v^m:$ $m\in \mathcal{M}\}$;
$E(H_{MEM})$ $=$ $\big\{$$(v^m, v^{m'}):$  $m\prec m'$ and $m,m'\in \mathcal{M}$$\big\}$.  The weight of every vertex $v^m\in V(H_{MEM})$ is defined as $W(v^{m})$ $=$ $|m_1|$.

\begin{lemma}
    \label{H-MEM-is-DAG}
    $H_{MEM}$ is a DAG.
\end{lemma}
\begin{proof}
    Suppose on the contrary that there is a cycle, $C$ $=$ $v^{m^1}$, $v^{m^2}$, $\dots$, $v^{m^k}$, $v^{m^1}$, in $H_{MEM}$. Then we have $m^{1}$ $\prec$ $m^{2}$ $\prec$ $\dots$ $\prec$ $m^k$ $\prec$ $m^1$.

    Let us consider $m^i_2$ for $1\leq i\leq k$, where the notations $m_0$, $m_1$ and $m_2$ are all defined in Section \ref{problems}. Let $m^1_2$ $=$ $[j_1,\dots,j'_1]$, $m^2_2$ $=$ $[j_2,\dots,j'_2]$, $\dots$, $m^k_2$ $=$ $[j_k,\dots,j'_k]$. Since $m^{1}$ $\prec$ $m^{2}$ $\prec$ $\dots$ $\prec$ $m^k$ $\prec$ $m^1$, we have $j'_1<j_2$, $j'_2<j_3$, $\dots$, $j'_{k-1}<j_k$, $j'_k< j_1$. It is obviously that  $j_2\leq j'_2$, $\dots$, $j_k\leq  j'_k$. So we have $j'_1<j_2$ $\leq $ $j'_2<j_3$ $\leq $ $\dots$ $\leq $ $j'_{k-1}<j_k$ $\leq $ $j'_k < j_1$. This contradicts to the fact that $j_1\leq j'_1$.

    Therefore, $H_{MEM}$ is a DAG.
\end{proof}

\begin{lemma}
    \label{path-to-SO-subset}
    Every path in $H_{MEM}$ indicates a strictly ordered subset of $\mathcal{M}$.
\end{lemma}
\begin{proof}
    Let $v^{m^1}$, $v^{m^2}$, $\dots$, $v^{m^k}$ be a path in $H_{MEM}$. According to the construction of $H_{MEM}$, we have $m^{1}$ $\prec$ $m^{2}$ $\prec$ $\dots$ $\prec$ $m^k$. Therefore, $\{$ $m^i:$ $1\leq i\leq k$ $\}$ is a strictly ordered subset of $\mathcal{M}$.
\end{proof}

\begin{lemma}
    \label{SO-subset-to-path}
    Every strictly ordered subset of $\mathcal{M}$ indicates a path in $H_{MEM}$.
\end{lemma}
\begin{proof}
    Let $\{$ $m^i:$ $1\leq i\leq k$ $\}$ be a strictly ordered subset of $\mathcal{M}$. Without loss of generality, let $m^{1}$ $\prec$ $m^{2}$ $\prec$ $\dots$ $\prec$ $m^k$. Then for every $1\leq i\leq k-1$, there is an edge from $v^{m^{i}}$  to $v^{m^{i+1}}$ in $H_{MEM}$. So the sequence $v^{m^1}$, $v^{m^2}$, $\dots$, $v^{m^k}$ is a path in $H_{MEM}$.
\end{proof}

\begin{theorem}
    \label{MEM-solution}
    A vertex-weighted longest path in $H_{MEM}$ indicates a strictly ordered subset $\mathcal{M}'$ 
    of $\mathcal{M}$ with $len(\mathcal{M}')$ maximized.
\end{theorem}
\begin{proof}
    By Lemmas \ref{path-to-SO-subset} and \ref{SO-subset-to-path}, paths in $H_{MEM}$ and strictly ordered subsets of $\mathcal{M}$ are in a one-to-one correspondence. In addition, recalling the construction of $H_{MEM}$, the weight of each vertex  $v^m \in V(H_{MEM})$ is the length of $m_1$.  So a vertex-weighted longest path in $H_{MEM}$ indicates a strictly ordered subset $\mathcal{M}'$ of 
    $\mathcal{M}$ 
    with $len(\mathcal{M}')$ maximized.
\end{proof}

Now Let us consider MSP problem. The instance of MSP problem is a query sequence $Q$, a pangenome graph $G$ and a set of seeds, $\mathcal{S}$, between them. We construct a vertex-weighted DAG $H_{MSP}$ as following: $V(H_{MSP})$ $=$ $\{$ $v^s: s\in \mathcal{S}$ $\}$; $E(H_{MSP})$ $=$ $\big\{$ $(v^s, v^{s'}):$ $s\prec s'$ and $s,s'\in \mathcal{S}$$\big\}$. The weight of  every vertex $v^s\in V(H_{MSP})$  is defined as $W(v^{s})$ $=$ $1$.

\begin{theorem}\label{H-MSP-solution}
    A vertex-weighted longest path in $H_{MSP}$ indicates a strictly ordered subset 
    of $\mathcal{S}$ with maximum seeds.
\end{theorem}
\begin{proof}
    The vertex set and the edge set of $H_{MSP}$ are established in the same way as those of $H_{MEM}$. While the MSP asks to find maximum number of seeds, it is equivalent to find a longest vertex-weighted path in $H_{MSP}$. So the correctness  follows directly from Theorem \ref{MEM-solution}.
\end{proof}

Now that we have reduced the four problems---LCS-SG, FGLCS-SG, MEMC and MSP--to the longest path problem in DAGs. A detailed and comprehensive analysis of the time complexity for each will be provided in the following section.

\section{5. Complexity analysis}
\label{complexity-analysis}
By Theorems \ref{H-longest-path-is-LCS} and \ref{H-gap-longest-path-is-LCS}, the optimal solutions to LCS-SG and FGLCS-SG correspond to edge-weighted longest paths  in $H$ and $H_{gap}$, respectively. Similarly, according to Theorems \ref{MEM-solution} and \ref{H-MSP-solution}, the optimal solutions to MEMC and MSP are vertex-weighted longest paths in $H_{MEM}$ and $H_{MSP}$, respectively. Section \ref{LP-DAG} demonstrates that both the edge-weighted and vertex-weighted longest path problems in DAGs can be solved in linear time. In the following, we provide a detailed time-complexity analysis for the constructions of $H$, $H_{gap}$, $H_{MEM}$, and $H_{MSP}$, respectively.

Given a sequence and a pangenome graph $G$, constructing $V(H)$ involves determining whether each $Q_i$ matches $\delta(u)_{f}$ for some $u \in V(G)$ and some $0 \leq f \leq |\delta(u)|-1$, where  $0 \leq i \leq |Q|-1$. Thus, the time complexity of building $V(H)$ is $O\big(|Q|N\big)$, where $N$ denotes the total sequence length across all vertices in $G$, specifically $N = \sum_{v \in V(G)} |\delta(v)|$.

The construction of $E(H)$ involves determining whether an edge exists between any pair of vertices in $V(H)$. Recall that $|V(H)| = O\big(|Q|N\big)$. Let $v_{(i,u,f)}$ and $v_{(i',u',f')}$ be two vertices in $V(H)$. When $u = u'$, it suffices to compare $i$ with $i'$ and $f$ with $f'$, which can be done in $O(1)$ time. When $u \neq u'$, it is necessary not only to compare $i$ and $i'$, but also to examine the reachability between $u$ and $u'$. Therefore, prior to constructing $E(H)$, we must determine the reachability for every pair of vertices in $V(G)$. As is well known, this can be accomplished in $O(n^3)$ time using the Floyd-Warshall algorithm for all-pairs shortest paths~\cite{cormen2022introduction}, where $n$ denotes the number of vertices in the graph. Consequently, the overall time complexity of constructing $E(H)$ becomes $O\big(n^3 + |Q|^2 N^2\big)$, where $O\big(|Q|^2 N^2\big)$ is the total number of vertex pairs in $V(H)$.

Given a sequence $Q$, a pangenome graph $G$, and two additional integers $k_1, k_2$, the construction of the vertex set $V(H_{gap})$ mirrors that of $V(H)$. We now analyze the time complexity involved in constructing $E(H_{gap})$. Consider two vertices $v_{(i,u,f)}$ and $v_{(i',u',f')}$ in $V(H_{gap})$. When $u = u'$, it suffices to verify whether $0 < i' - i \leq k_1$ and $0 < f' - f \leq k_2$, an operation that requires only $O(1)$ time.

If $u \neq u'$, it is necessary not only to verify that $0 < i' - i \leq k_1$, but also to ensure $0 < \text{ind}\big(\delta(u')_{f'}, \text{spell}(P)\big) - \text{ind}\big(\delta(u)_f, \text{spell}(P)\big) \leq k_2$. To address this, we convert the given pangenome graph into a one-character graph by splitting  every vertex $v \in V(G)$ into $|\delta(v)|$ individual vertices while preserving the origin of each. The resulting graph is denoted as $G^s$. The order of $G^s$ is just equal to $N = \sum_{v \in V(G)} |\delta(v)|$, and its size amounts to $|E(G)| + N - n$. To verify whether there is a path $P$ from a certain vertex $u\in V(G)$ to a certain vertex $u'\in V(G)$ such that $0 < \text{ind}\big(\delta(u')_{f'}, \text{spell}(P)\big) - \text{ind}\big(\delta(u)_f, \text{spell}(P)\big) \leq k_2$ for some $0\leq f < |\delta(u)|$ and some $0\leq f' < |\delta(u')|$, it suffices to check whether the distance between the vertex labeled $\delta(u)_f$ and the vertex labeled $\delta(u')_{f'}$ in the updated graph $G^s$ does not exceed $k_2$. This meets with all-pairs shortest paths problem  in the graph $G^s$, which can be addressed by  the Floyd-Warshall algorithm and takes  $O(N^3)$ time \cite{cormen2022introduction}. Consequently, constructing $E(H_{\text{gap}})$ takes $O\big(N^3 + |Q|^2 N^2\big)$ time.

Given a set of MEMs, $\mathcal{M}$, and a pangenome graph, $G$, the vertex set $V(H_{MEM}) = \{v^m : m \in \mathcal{M}\}$ can be constructed in $O\big(|\mathcal{M}|\big)$ time. We now turn our attention to the time complexity of constructing $E(H_{MEM})$. Consider two distinct MEMs, $m\in$ $\mathcal{M}$ and $m'\in$ $\mathcal{M}$. Determining whether an edge exists from $v^m$ to $v^{m'}$ amounts to checking if $m \prec m'$. Recall the definition of "$\prec$" in Section \ref{problems}. When $m_0 = m'_0$, it suffices to compare the last element of the list $m_1$ with the first element of the list $m'_1$, which can be done in $O(1)$ time. In the case $m_0 \neq m'_0$, it suffices to checking reachability from $m_0$ to $m'_0$ in the graph $G$, thereby transforming it into a standard graph reachability problem. The Floyd-Warshall algorithm requires $O(n^3)$ time to solve the reachability problem for all vertex pairs in $G$ \cite{cormen2022introduction}, where $n = |V(G)|$. Consequently, constructing $E(H_{MEM})$ takes $O\big(n^3 + |\mathcal{M}|^2\big)$ time. 

Finally, the construction process for $H_{MSP}$ mirrors that of $H_{MEM}$, resulting in a time complexity of $O\big(n^3 + |\mathcal{S}|^2\big)$, where $n$ denotes the number of vertices in the given pangenome graph, and $\mathcal{S}$ represents the seed set in the MSP instance.

\section{6. Conclusions and discussions}
\label{conclusion}

This paper investigates the longest common subsequence problem and its three other variants between a sequence and a pangenome graph, offering a comprehensive analysis of their computational properties and implications. We demonstrate that all four problems can be solved in polynomial time by reducing them to longest path problems in directed acyclic graphs (DAGs), highlighting their computational tractability through well-established graph-theoretic methods. 

However, the cubic time complexity proves highly impractical for large-scale instances, as it leads to excessive computational time and resource consumption, thereby posing significant challenges in real-world applications.  With the continuous expansion of biological sequence data across diverse domains, the drawbacks of such high time complexity become ever more apparent, impeding prompt and efficient problem resolution. To overcome this limitation, future research should prioritize the design and development of more efficient algorithms with sub-cubic or even sub-quadratic time complexity, paving the way for scalable and practical solutions to complex computational challenges. Sub-cubic time complexity denotes an algorithm whose runtime increases at a rate slower than a cubic function $O(n^3)$  as the input size $n$ grows.

{\footnotesize

\bibliographystyle{plain}
\bibliography{sample-base}

}

\end{document}